\documentclass[12pt]{article}
\usepackage{amsmath,amsthm,amscd,amsfonts,amssymb}
\usepackage{latexsym}
\usepackage{comment}
\usepackage[usenames]{color}
\newcommand{\CC}{\mathbb C}

\newcommand{\NN}{\mathbb N}
\newcommand{\PP}{\mathbb P}
\newcommand{\RR}{\mathbb R}

\newcommand{\ZZ}{\mathbb Z}
\newcommand{\EE}{\mathbb E}

\newcommand{\hh}{\mathcal H}

\newcommand{\kk}{\mathcal K}
\newcommand{\nn}{\mathcal N}

\newcommand{\dir}[1]{\mathbf{#1}}
\newtheorem{thm}{Theorem}[section]

\newtheorem{cor}[thm]{Corollary}
\newtheorem{prop}[thm]{Proposition}

\newtheorem{rem}[thm]{Remark}

\begin{document}
\author{M Krishna
 \\ Institute of Mathematical Sciences, \\
Taramani, Chennai 600113, India \\
krishna@imsc.res.in \\
\\
Peter Stollmann \\
Department of Mathematics \\
Technical University of Chemnitz \\
09107 Chemnitz, Germany \\
peter.stollmann@mathematik.tu-chemnitz.de}

\title{Direct integrals and spectral averaging}
\date{}
\maketitle
\begin{abstract}
A one parameter family of selfadjoint operators gives rise to a corresponding
direct integral. We show how to use the Putnam Kato theorem to obtain a new
method for the proof of a spectral averaging result.
\end{abstract}

\section{Introduction}
To us the basic issue of spectral averaging is to derive continuity properties
of an
integral of spectral measures; thus we  consider a selfadjoint operator $A$ in
a separable Hilbert space $\hh$ as well as a bounded operator $B$ on $\hh$,
$B\ge 0$ and denote $H(t):= A+tB$.  We write $\rho_{H(t)}^\Phi$ for the
spectral measure of $H(t)$ with respect to the vector $\Phi\in\hh$. Our main
result is

\begin{thm}\label{Thm1}
Let $H(\cdot)$ be as above and let $\Phi \in \overline{Range (B)}$.  Then
the
measures
$$
\nu  = \int \rho_{H(t)}^{\Phi} h(t) dt
$$
are absolutely continuous (with respect to Lebesgue measure) for any $h \in
L^1(\RR)$.
\end{thm}

Results of this type have quite a history and due to their importance for random
operators, the interest has been steady. We refer to
\cite{Kotani-84,SimonW-86,Howland-87}
and the
references in there for early results, partly building on even older work
\cite{Donoghue-65} and to \cite{CHK-03,CHK-07,Sto-10} for
the more recent state of matters.
Note however that we concentrate on one part of the intrigue, the continuity of
the integrated spectral measures, while the emphasis in the cited works is
somewhat different. There the main point is to deduce the spectral type of the
single operators $H(t)$ the integral is made of. Clearly, in the setting of our
main result nothing can be said about that.

The main improvement that had happened during the last 20 years of
development is the generality of the operator $B$ that appears, a feature that
is of prime importance for applications to random operators. One of the
main ideas that enter the usual proof, as presented, e.g. in \cite{Sto-10}, has
also been fundamental in the adaptation of the fractional moment method to
continuum
models, cf. \cite{AENSS-06}. It uses the fact that a maximally accretive
operator can always be obtained as the dilation of a selfadjoint operator. In
contrast, in the early  papers $B$ was merely a rank one projection which
already turned out to be extremely useful for discrete random models.

Our proof of the above theorem is quite different: we consider
$$
\dir{H}:=\int^\oplus_{\RR} H(t)dt \mbox{  in  }\int^\oplus_{\RR} \hh dt
$$
and apply the Kato-Putnam theorem to this operator to show that some of its
spectral measures are absolutely continuous. (In the next section we recall the
necessary notions from the theory of direct integrals of Hilbert spaces.) We
should like to point out that the idea to apply Mourre theory to obtain
spectral averaging results can be found in \cite{CHM-96}, leading to a somewhat
different proof that is nevertheless quite related to what we have
done here. A major point in the present paper is the simplicity of the method.

\section{Spectral averaging and direct integrals}\label{sec2}
What we need about direct integrals can be found in \cite{RS-4}, p. 280 ff.

As we remarked above, we are dealing with a separable Hilbert space $\hh$ and
consider the constant fibre direct integral

$$
\kk = L^2(\RR , \hh) = \int^{\oplus}_{\RR} \hh dt,
$$
with the inner product $\langle f, g \rangle_{\kk} =
\int \langle \overline{f(t)}, g(t)\rangle_{\hh}  ~ dt $.
The direct integral of a selfadjoint operator function is described in:
\begin{rem}\label{Rem}
Let $H(t)$ be selfadjoint in $\hh$ for $t\in\RR$. Then $D(H):=\{ f\in L^2(\RR,
\hh)\mid f(t)\in D(H(t))\mbox{  for a.e.  }t\in\RR, \int_{\RR}\| H(t)f(t)\|^2
dt<\infty \}$, $Hf:=\int^{\oplus}_{\RR}H(t)f(t)dt$ defines a selfadjoint
operator. It follows that $\phi(H)$ is decomposable for any bounded measurable
$\phi:\RR\to\CC$ and
$$
\phi(H)=\int^{\oplus}_{\RR}\phi(H(t))dt .
$$
In particular,
$$
\langle E_H(I)f\otimes g,f\otimes g\rangle =\int_{\RR} \langle
E_{H(t)}(I)f,f\rangle | g(t)|^2dt
$$
for the spectral projections and
$$
\rho^{f\otimes g}_{H}=\int_{\RR} \rho^{f}_{H(t)}| g(t)|^2dt
$$
for the spectral measures.
\end{rem}
See \cite{RS-4}, p. 280 ff, in particular Thm XIII.85. The latter formula makes
the connection to spectral averaging clear.

Note that the obvious isometric isomorphism gives
$$
\kk = \hh\otimes L^2(\RR) .
$$
We will use this additional structure and write, e.g.
$$
\dir{A}:=A\otimes 1
$$
for the canonical extension of $A$ (which is a selfadjoint operator in $\hh$)
to $\kk$. In much the same way we extend the position operator $Q$. Using some
ideas from \cite{Howland-87} we introduce the following: $T=\tanh Q$ the
maximal multiplication operator in $L^2(\RR)$ with $\tanh$, as well as $D:=
\arctan(P)$, where $P=-i\frac{d}{dt}$ is the momentum operator in  $L^2(\RR)$.

\begin{prop} \emph{(\cite[Lemma 2.9]{Howland-87})} On $L^2(\RR)$,
consider the
operators $T$ and $D$ above. Then $i[T,D]=C$ is positive definite.
\end{prop}
We next infer the following result of Putnam and Kato \cite{Kato-68,RS-4}:

\begin{prop}\label{putnam}
 Let $H$ and $D$ be selfadjoint and $D$ be bounded. If $C=i[T,D]\ge 0$, then
$H$ is absolutely continuous on $Range(C)$.
\end{prop}

 \begin{cor}
  The operator $\hat{\dir{H}}=\int^\oplus_{\RR}(A+\tanh t B)dt$ is absolutely
continuous on $\overline{Range (B)}\otimes L^2(\RR)$.
 \end{cor}
\begin{proof}
 By what we know from above,
$$
i[\hat{\dir{H}},\dir{D}]=\dir{B}\dir{C}=B\otimes C\ge 0.$$
Since $C$ is positive definite it follows that $Range(C)$ is dense in
$L^2(\RR)$.
\end{proof}
\begin{proof}[\textbf{Proof of Theorem \ref{Thm1}}]
 \textbf{Step 1:} The preceding Corollary and the above Remark \ref{Rem}
give that for any $\Phi\in \overline{Range (B)}$, $g\in L^2$,
$$
\int \rho_{A+\tanh t B}^{\Phi} |g(t)|^2 dt<<dt,
$$
where the latter indicates absolute continuity with respect to Lebesgue measure.

\noindent \textbf{Step 2:} By specializing and change of variables:
 For any $\Phi\in \overline{Range (B)}$, $g\in L^\infty$ with compact support:
$$
\int \rho_{A+ t B}^{\Phi} |g(t)|^2 dt<<dt,
$$
Now by approximation, we get arbitrary positive $h\in L^1$ and, by linearity,
the assertion of the Theorem.
\end{proof}
A standard extension formulated in a way that is suited for the application we
have in mind is the following Corollary from the proof of Theorem \ref{Thm1}.
\begin{cor}\label{cor1.2}
 Let $A$ and $B$ be as above and assume that $\overline{\{ \varphi(A)B f\mid
f\in \hh\}}=\hh$. Then, for any  $h\in L^1$ and any $\phi\in\hh$:
$$
\int_{\RR} \langle
E_{A+tB}(\cdot)\phi,\phi\rangle h(t)dt<<dt .
$$
\end{cor}
\begin{proof}
We consider the operators $\hat{\dir{H}}$, $\dir{A}$ and $\dir{B}$ on $\kk$ as
above. By what we proved above, the absolutely continuous subspace of
$\hat{\dir{H}}$ contains $Range(B)\otimes L^2(\RR)$. Moreover it is cyclic for
$\hat{\dir{H}}$ and closed. Therefore, the arguments from \cite{Howland-87},
Proof of Theorem 2.7, p.61 give that the absolutely continuous subspace of
$\hat{\dir{H}}$ is all of $\kk$. As in the above proof this implies the asserted
absolute continuity.
\end{proof}

It is time to compare what we have shown so far with what is known by other
methods, see \cite{CHM-96,CHK-03,CHK-07, Sto-10}.

\begin{rem}
 \begin{itemize}
  \item Strictly speaking, the results of \cite{CHK-03,CHK-07, Sto-10} and our
Corollary above are not comparable, but the latter can be used to deduce what
we have shown here. More precisely:
\item In \cite{CHK-03,CHK-07, Sto-10} instead of $h(t)dt$ more general measures
are allowed. The continuity of $\nu:=\int_{\RR} \langle
E_{A+tB}(\cdot)\phi,\phi\rangle d\mu(t)$ as well as that of $\mu$ is measured in
terms of the modulus of continuity $s(\mu,\varepsilon):=\sup\{\mu([a,b])\mid
a,b\in\RR, b-a=\varepsilon\}$ and the conclusion is that $s(\nu,\varepsilon)\le
C s(\mu,\varepsilon)$, provided $\phi\in Range(B^\frac12)$.
\item Clearly, the latter estimate directly does not give anything in the case
of our result above: for absolutely continuous $\mu=h(t)dt$ the modulus of
continuity does not  need to decay at a certain rate as $\varepsilon$ tends to
zero. But, we can approximate $h$ by bounded $h_n$ in a suitable way. At the
same time, we can approximate any $\phi\in \overline{Range(B)}$ by a sequence
$\phi_n\in Range(B^\frac12)$ the resulting $\nu_n:= \int_{\RR} \langle
E_{A+tB}(\cdot)\phi_n,\phi_n\rangle h_n(t)dt$ will converge to $\nu$ and all
the $\nu_n$ are absolutely continuous with respect to $dt$, thus giving the
assertion of our Theorem \ref{Thm1}.
\item In \cite{CHM-96} the method of proof is pretty much similar to our
strategy here. There, a direct application of Mourre estimates is used to prove
spectral averaging and Wegner estimates. While their result concerns a more
general setup it requires even differentiability of the density $h$; see Thm
1.1, Cor. 1.2 and 1.4 in the cited paper for results analogous to ours.

 \end{itemize}
\end{rem}

\section{Absolute continuity of the IDS; a very short proof.}

We consider $L^2(\RR^d)$ and the operators
\begin{equation}\label{eqn9}
H^\omega = -\Delta + \sum_{n \in \ZZ^d} \omega_n u(\cdot -n)
\end{equation}
where $u$ is a non-negative bounded measurable function that is positive on some
open set.
Let $\omega_n, n\in \ZZ^d$ be i.i.d random variables with a
probability distribution
$\mu$ which is absolutely continuous and has a compactly
supported, integrable density $h$. We denote by $\PP:=\bigotimes_{n\in\ZZ}\mu$
the product measure and by $\EE$ the corresponding expectation.

By $\Lambda(0)$ we denote the unit unit cube. In view of the Pastur-Shubin trace
formula we can express the integrated density of states, IDS, in terms of
\begin{equation}\label{dos}
\nn(I) = \EE \left[ Tr(\chi_{\Lambda(0)} E_{H^\omega} (I) \chi_{\Lambda(0)})
\right],
\end{equation}
for any bounded Borel set $I$. See \cite{Ves-08} for an extensive bibligraphy
on the IDS and the proof of the trace formula in a more general situation. The
IDS is
quite often also expressed as the distribution function $N(E):\nn(-\infty,E]$
of the measure $\nn$ defined above.
\begin{cor}
 In the situation above, $\nn<<dt$.
\end{cor}
\begin{proof}
 Note that in the situation given we can apply the cyclicity result of
\cite{CH-94}, Prop. A2.2, and know that for
$$
A:= -\Delta + \sum_{n \in \ZZ^d\setminus \{ 0\}} \omega_n u(\cdot -n)
\qquad B:= u(\cdot),
$$
the assumptions of Corollary \ref{cor1.2} are met. We fix an orthonormal basis
$(\phi_k)_{k\in\NN}$ of $\hh=L^2(\RR^d)$ and write
\begin{eqnarray*}
\nn(I) &=& \EE \left[ Tr(\chi_{\Lambda(0)} E_{H^\omega} (I)
\chi_{\Lambda(0)})
\right]\\
&=& \EE \left[\sum_{k\in\NN}\langle\chi_{\Lambda(0)} E_{H^\omega} (I)
\chi_{\Lambda(0)}\phi_k\mid\phi_k\rangle\right]\\
&=& \sum_{k\in\NN}\EE \left[ E_{H^\omega} (I)
\chi_{\Lambda(0)}\phi_k\mid\chi_{\Lambda(0)}\phi_k\rangle\right] .
\end{eqnarray*}
It suffices to show that every sum in the term is absolutely continuous with
respect to $dt$ and this works as follows:
$$
 \EE \left[ E_{H^\omega} (I)
\chi_{\Lambda(0)}\phi_k\mid\chi_{\Lambda(0)}\phi_k\rangle\right]=$$
$$=
\EE \left[ E_{(-\Delta + \sum_{n \in \ZZ^d\setminus \{ 0\}} \omega_n u(\cdot
-n)+\omega_0 B)} (I)
\chi_{\Lambda(0)}\phi_k\mid\chi_{\Lambda(0)}\phi_k\rangle\right] =$$
$$=
\EE \left[\int_\RR E_{(-\Delta + \sum_{n \in \ZZ^d\setminus \{ 0\}} \omega_n
u(\cdot
-n) + t B)} (I)
\chi_{\Lambda(0)}\phi_k\mid\chi_{\Lambda(0)}\phi_k\rangle dt\right] .
$$
For fixed $\omega':= (\omega_n)_{n\in\ZZ^d\setminus \{ 0\}}$ the inner integral
is seen to give an absolutely continuous measure: set $A$ as above and apply
Cor. \ref{cor1.2}. The expectation preserves the absolute continuity and that
establishes the claim.
\end{proof}

Of course, an additional periodic background potential $V_0$ would not change
the proof; all the ingredients we cited are valid in this case as well.
\paragraph*{Acknowledgement.}
 This work was started during a visit of P.S. to India; he wants to thank for
generous support by the IMSC, Chennai and the hospitality of the staff there as
well as for the generous support by the DFG (German Science Foundation).

\end{document}